\documentclass{iopart}
\usepackage{amsfonts} 
\usepackage[utf8]{inputenc}
\usepackage{xcolor}
\usepackage{hyperref}
\usepackage{graphicx,amsthm}
\usepackage[legalpaper,margin=1.in]{geometry}
\usepackage{iopams}

\definecolor{Gr}{HTML}{377D71}
\definecolor{Pu}{HTML}{A459D1}
\definecolor{Bl}{HTML}{4D455D}
\definecolor{Te}{HTML}{C1ECE4}
\definecolor{Or}{HTML}{EF6262}
\definecolor{Am}{HTML}{F3AA60}
\definecolor{Co}{HTML}{3C486B}
\definecolor{Wh}{HTML}{FEFBF6}
\definecolor{Ye}{HTML}{FFE196}
\definecolor{Re}{HTML}{E96479}
\definecolor{Pi}{HTML}{FFD0D0}
\definecolor{Rp}{HTML}{FF9EAA}
\definecolor{Wg}{HTML}{EEF3D2}

\newcommand{\Rv}[1]{#1} 

\hypersetup{colorlinks=true,linkcolor=red,citecolor=red}
\renewcommand*\thesubsection{\arabic{section}.\arabic{subsection}}
\usepackage{titlesec}
\titleformat{\section}
  {\normalfont\sffamily\Large\bfseries}
  {\thesection}{1em}{}
  \titleformat{\subsection}
  {\normalfont\sffamily}
  {\thesubsection}{1em}{}

\usepackage{etoolbox}
\makeatletter
\newrobustcmd{\fixappendix}{%
  \patchcmd{\l@section}{1.5em}{7em}{}{}%
  \patchcmd{\l@subsection}{2.3em}{7em}{}{}%
}
\makeatother

\newtheorem{thm}{Theorem}
\newtheorem{lem}{Lemma}

\newtheorem{prop}{Proposition}

\newcommand{\mcC}{\mathcal C}

\newcommand{\mcP}{\mathcal P}
\newcommand{\mcN}{\mathcal N}

\newcommand{\reals}{\mathbb R}
\newcommand{\ints}{\mathbb N}
\newcommand{\naturals}{\mathbb Z}

\newcommand{\bbP}{\mathbb P}
\newcommand{\E}{\mathbb E}
\newcommand{\noiseparam}{\zeta} 
\newcommand{\growthfactor}{\Phi} 

\begin{document}

\title{\sffamily Coalescent processes emerging from large deviations}
\author{Ethan Levien}

\address{Dartmouth College, Department of Mathematics}

\ead{ethan.a.levien@dartmouth.edu}


\begin{abstract}
The classical model for the genealogies  of a neutrally evolving population in a fixed environment is due to Kingman. Kingman's coalescent process, which produces a binary tree, universally emerges  from many microscopic models in which the variance in the number of offspring is finite. It is understood that power-law offspring distributions with infinite variance can result in a very different type of coalescent structure with merging of more than two lineages. Here we investigate the regime where the variance of the offspring distribution is finite but comparable to the population size. This is achieved by studying a model in which the log offspring sizes have a stretched exponential form.  Such offspring distributions are motivated by biology, where they emerge from a toy model of growth in a heterogenous environment, but also mathematics and statistical physics, where limit theorems and phase transitions for sums over random exponentials have received considerable attention due to their appearance in the partition function of Derrida's Random Energy Model (REM). We find that the limit coalescent is a $\beta$-coalescent -- a previously studied model emerging from evolutionary dynamics models with heavy-tailed offspring distributions. We also discuss the connection to previous results on the REM. \end{abstract}

\section{Introduction}

Evolution is, in large part, shaped by a tension between two opposing ``forces'': neutral genetic drift and selection. Neutral genetic drift refers to random changes in the genetic composition of a population due to chance events in the deaths and reproduction of individuals. Selection, in contrast, results from a deterministic bias towards fitter individuals. A major objective of evolutionary biology is to determine the relative impact of these forces. In order to achieve this, we need to understand how different microscopic mechanisms manifest in macroscopic observables, such that the frequency of a genotype or the shape of the genealogical tree. 

Much of our understanding of the interplay between genetic drift and selection comes from the Wright-Fisher model \cite{neher2018progress,durrett2008probability}, or its counterpart with overlapping generations, the Moran model. In both models, the source of noise is the random sampling of individuals from generation to generation. In the large population size limit this sampling noise has a variance which is inversely proportional to $N$, the population size. When genetic drift is the sole source of changes in the composition of the population -- that is, in neutrally evolving populations --  $N$ sets the characteristic time-scale of the evolutionary dynamics.

The true population size is rarely related in a simple way to the variance in genotype frequency, as predicted by the Wright-Fisher model. Therefore, one usually thinks of an $N$ as an \emph{effective population size} measuring the overall strength of genetic drift, rather than the literal number of cells. This has motivated the question: What determines the effective population size? Of particular interest is the question of how non-genetic variation between individuals shapes the effective population size \cite{jafarpour2022evolutionary,levien2021non}. J.H. Gillespie was one of the first to address this question by considering a model in which the number of offspring from each individual is a random variable having finite variance. He showed that in this context the effective population size is obtained by scaling the true population size by the variance in offspring \cite{gillespie1973natural,starrfelt2012bet,schreiber2015unifying}.

We now have a more general and mathematically rigorous understanding of this problem which is based on the Cannings model \cite{cannings1974latent}.  Beginning with $N$ labeled cells, trajectories of the Cannings model are constructed by generating an exchangeable random vector $(\nu_{1,k},\dots,\nu_{N,k})$ satisfying $\sum_{i=1}^N \nu_{i,k}=N$ for each $k \in \ints$.  $\nu_{i,k}$ represents the number of offspring in generation $k$ which descend from individual $i$ in the $(k-1)$th time-step.  We will henceforth omit the subscript $k$ and it should be understood that all quantities associated with a generation are implicitly dependent on $k$.\footnote{The term \emph{generation} is used to refer to a time-step in the Cannings model.}
 Inspired by  \cite{schweinsberg2003coalescent,hallatschek2018selection,okada2021dynamic}, we will focus on the particular case where $\nu_i$ is obtained by first generating iid random variables $\{U_i\}_{i=1}^N$ representing the number of offspring of each individual \emph{before} resampling, and then sampling the resulting offspring pool, without replacement, to obtain the individuals in the next  generation.  Conditional on the offspring numbers $U_1,\dots,U_N$, $\nu_i$ follows a hypergeometric distribution:
 \begin{equation}
\bbP\left(\nu_i = n|\{U_i\}_{i=1}^N\right) = \frac{1}{{S_N \choose N}}{U_i \choose n}{S_N - U_i \choose N - n}.
\end{equation}
 This formulation bears a close resemblance to the model studied by Gillespie \cite{gillespie1973natural} and is an example of a Generalized Wright-Fisher model defined in \cite{der2011generalized}.

One limit theorem for the Cannings model concerns the behavior of genotype frequencies over long time-scales in large populations.  By long time-scales, we mean on the order of the \emph{coalescent time}, defined as the average number of generations we must travel backwards in time to find a common ancestor for two randomly selected individuals in the same generation.  The coalescent time is equivalent to the effective population size in some cases, but is more general in the sense that the definition does not require a mapping to the Wright-Fisher model.   Under the assumption that the variance in the number of offspring produced by each individual is finite, the time-rescaled genotype frequencies converge as $N \to \infty$ (in the  Skorokhod sense  -- see \cite{kern2022skorokhod,ethier2009markov}) to the well-known Wright-Fisher diffusion (WFD).  \cite{mohle2001classification,gillespie1974natural}. In the WFD, the change in the frequency, $X(t)$ of a genotype over a time interval $dt\ll 1$ is Normally distributed with mean zero and variance $X(t)(1-X(t))dt$. 

There is similar limit Theorem, due to Kingman \cite{kingman1982genealogy}, for the genealogical trees in the Cannings model. These trees can be generated by the stochastic process known as a \emph{coalescent processes} which, roughly speaking, is specifies which lineages have merged $k$ generations back in time from our original sample.
 Under the condition that the genotype frequencies converge to the WFD, the time-rescaled coalescent process converges to a continuous time process known now as the Kingman Coalescent.  The genealogical tree produced by this process is almost surely a binary tree where pairs of lineages merge at a rate one, and importantly, the probably of more than two lineages merging in a single instant is zero. Similar results exist for other microscopic models (e.g. the Moran process) and a large body of work focuses on understanding how the microscopic details of the process shape the coalescent time, or effective population size \cite{gillespie1974natural,jafarpour2022evolutionary,charlesworth2009effective,gillespie2004population,gillespie2001population}.



The WFD/Kingman models are not always sufficient to capture the dynamics of real evolution. This is due to, for example, \emph{multiple merger} coalescents appearing in experimental data \cite{tellier2014coalescence,sackman2019inferring}. In multiple merger coalescents there is a non-negligible probability to observe more than two individuals sharing a common ancestor in a single unit of time (e.g. a generation in the model) on time-scales of order of the coalescent time.  In \cite{schweinsberg2003coalescent}, Schweinsberg explored the question of whether such genealogies can emerge from neutral evolution by studying the Cannings model with power-law distributions.   A number of papers have since investigated the role of power-law tail offspring distribution in generating multiple merger coalescents and non-diffusive genotype frequency fluctuations \cite{hallatschek2018selection,okada2021dynamic,cordero2022lambda}. The genotype frequency dynamics which emerge from power law offspring -- known as $\Lambda$-Fleming Viot processes -- are non-diffusive processes which, in general,  have discontinuous sample paths.  With the notable exception of \cite{cordero2022lambda}, most previous work has focused on how multiple merger coalescents emerge when the variance in offspring is infinite. The assumption of infinite variance is a mathematical convenience, which may be justified in some cases e.g. for certain models of dormancy \cite{cordero2022lambda,wright2019stochastic} and rare mutations \cite{luria1943mutations}. What remains unclear is which coalescent processes emerge when the offspring distribution has a variance which is finite, yet large enough relative to the population size to give rise to multiple merger coalescents.

In this paper,  in order to understand the role of large but finite offspring variability, we investigate the limit processes which emerge when the population size and offspring variability are simultaneously taken to be large. Similar limits were investigated in  \cite{cordero2022lambda}, but we focus on a more specific scaling between population size and offspring variability, which allows us to obtain precise descriptions of the limit coalescents.  The offspring distribution we consider and our scaling assumption are both inspired by prior work on the Random Energy Model (REM) of disordered systems.

Our main result (Theorem \ref{thm:main}) says that the limit processes emerging form the genealogies of the Cannings model under this scaling limit, called $\beta$-coalescents, are the same as the coalescent processes emerging from power law offspring. Our model is parameterized by a scaling rate which is analogous to the temperature of the REM. Just as in the REM, we find there are two critical points. Below the lower critical point there is no continuous time limit process, while between the two critical points one finds multiple merger coalescents. However, while the lower critical point corresponds exactly to the lower critical point of the REM marking the transition to the ``frozen'' state, the upper critical point does have the obvious interpretation in the context of the REM (which would be to separate the regimes of strong and weak-self averaging of the partition function). Our results complement and expand upon the existing connection between coalescent theory and the REM, which was made by Bolthausen and Sznitman in \cite{bolthausen1998ruelle}.


\subsection{Organization of this paper}
This paper is organized as follows. In Section \ref{sec:model} we describe the model under consideration, which is a particular instance of the Cannings model.    In Sections \ref{sec:background_col} and  \ref{sec:powerlaw} we review some background of coalescent theory and review what is known about this model. In Sections \ref{sec:largebutfinite} we present our main result which concerns the limiting coalescent when both the population size and offspring variation are taken to be large. We also discuss the relationship between our results and \cite{cordero2022lambda}. Section \ref{sec:rem} is devoted to the Random Energy Model and the connection between our results and the thermodynamic limit of the REM.  

\section{Background}

\Rv{Throughout this paper, we use the following standard notation.  $g \sim f$ means $f/g \to 1$,  $a\wedge b= \min\{a,b\}$,  $a \vee b = \max\{a,b\}$, $(n)_k = n(n-1)\cdots (n-k+1)$ and $[n] = \{1,2,\dots,n\}$. 
}

\subsection{Exponential offspring model}\label{sec:model}

 \Rv{In order to study the situation where variation in offspring numbers is finite, but large relative to the population size, we set
\begin{equation}\label{eq:Uiexp}
U_i = e^{\noiseparam \growthfactor_i},\quad i \in [N]
\end{equation}
where $\{\growthfactor_i\}_{i=1}^N$ are iid and $\noiseparam$ is a deterministic parameter.  Since we are interested in the large noise limit, we will eventually take $\noiseparam \to \infty$. In the remainder of this paper, in order to avoid ambiguity we will often replace $i$ with $1$ when referencing elements of an exchangeable random vector. 
 Note that, technically, $U_1$ should take values in $\naturals$, however in the large noise limit the distinction is not relevant. Moreover, taking $\noiseparam \to \infty$ implies $\E[U_1]\to \infty$, so we can assume that
\begin{equation}\label{eq:uass}
\E[U_1]>1. 
\end{equation}
Offspring distribution of the form \ref{eq:Uiexp} were also considered in \cite{Siri2023beyond} within the context of a model including heritability and \cite{cordero2022lambda} in a model of dormancy. }

\Rv{It is biologically sensible to work with variation on an exponential scale whenever the organisms in question proliferate exponentially between bottlenecks.  In this context, $\noiseparam \growthfactor_i$ is the product of the duration of growth between bottlenecks and the time-averaged growth rate for the offspring of the $i$th cell.  For example, imagine an asexual population that is subject to successive cycles of growth and dilution in a heterogeneous environment. An example is a microbial pathogen, such as \emph{Mycobacterium tuberculosis} \cite{gagneux2018ecology,freschi2021population}. Suppose that after passing through a bottleneck, each of the $N$ cells occupies a spatially distinct location. Then, due to heterogeneity in the environment, the offspring of each cell will proliferate at different rates, $\Phi_i$. If $\noiseparam$ is the duration of the growth phase, the total number of offspring from the $i$th cell (before resampling) will be of the form \ref{eq:Uiexp}.  Alternatively, one could imagine that it is the duration of the growth phase which is random -- for example, due to a period of dormancy before growth -- while the growth rates are fixed.  In this case we would use $\noiseparam$ to denote the growth rates and $\Phi_i$ the duration of the growth phase. Such a dormancy model was studied in \cite{cordero2022lambda}.  }

We focus on the case where $\growthfactor_1$ has a stretched exponential distribution, which essentially means the large deviations of $\growthfactor_1$ are scale-invariant. To be precise, let 
\begin{equation*}
W(\phi) \equiv -\ln P(\growthfactor_1 > \phi).
\end{equation*}
We assume $W(\phi)$  is smooth and bounded as $\phi \to 0$ and, inspired by \cite{eisele1983third,ben2005limit}, assume 
\begin{equation}\label{eq:hass}
W(\phi) \sim W^*(\phi) \equiv  \frac{1}{q}\phi^{q},\quad q\ge 1.
\end{equation}
as $\phi \to \infty$.  
The pre-factor of $1/q$ is arbitrary, since this could be absorbed into $\noiseparam$, but this particular form will lead to some more elegant analytical formula. Note that we can always set $\E[\growthfactor_1]=0$ without loss of generality, since the contribution from $e^{\noiseparam \E[\growthfactor_1]}$ cancels in the ratio $U_1/S_N$. Note that offspring distribution satisfying Equations \ref{eq:uass} and \ref{eq:hass} include the special cases where $U_1$ is exponential $(q=1)$ and lognormal $(q=2)$.

\subsection{Characterization of limit coalescent processes}\label{sec:background_col}

For a realization of the Cannings model with population size $N$, the corresponding discrete coalescent process on a sample of size $n$, denoted by $(\Psi_{n,k}^N)_{k\ge0}$, describes the genealogical tree obtained by following each sampled individual's ancestors back through time and grouping branches when individuals share a common ancestor. This process is shown in Figure \ref{fig:cannings}. We can intuitively understand $\Psi_{n,k}^N$ as the state of this tree $k$ generations back from the time our original sample was taken. To define $\Psi_{n,k}^N$ more mathematically, let $\mcP_n$ denote the space of partitions on $[n] = \{1,\dots,n\}$; that is, each element of $\mcP_n$ is a set of disjoint subsets of $[n]$ whose union is $[n]$. Then $\Psi_{n,k}^N$ is a Markov chain taking values in $\mcP_n$ and $\Psi_{n,0}^N$ is the partition into singletons; that is, $\Psi_{n,0}^N = \{\{1\},\{2\},\dots,\{n\}\}$. 
For $k>0$, indices $i,j \in [n]$ are in the same block of the partition if and only if the $i$th and $j$th individuals in the original sample share an ancestor $k$ generation in the past.

The continuous time coalescent processes which emerge as limits of $\Psi_{n,k}^N$ from any exchangeable Cannings model have been characterized in \cite{mohle2001classification}. The authors consider the large-$N$ limit of the time-scaled coalescent process 
\begin{equation*}
\Psi_{n}^N(t) = \Psi_{n,\lfloor t/c_N\rfloor}^N
\end{equation*}
where $c_N$ is the probability two random selected individuals in one generation share an ancestor in the previous generation.  Observe that $c_N^{-1}$ is simply the coalescent time, since the number of generations to the most recent common ancestor of two random selected individuals from the current generation follows a geometric distribution with parameter $c_N$. Conditional on $\{\nu_i\}$, the chance for two individuals to both be descendants of the first individual in the previous generation is $(\nu_1/N)(\nu_1-1)/(N-1)$. Multiplying by $N$ and averaging over all possible $\{\nu_i\}_{i=1}^N$ gives 
\begin{equation}\label{eq:cNdef}
c_N = \frac{\E\left[\nu_1(\nu_1-1)\right]}{N-1}. 
\end{equation}
We emphasize that, given the distribution of $U_i$, it is not straightforward to compute $c_N$ because it has a nonlinear dependence on the sum $S_N$.

The precise notion of convergence considered is in the Skorohod sense \footnote{Let ${\mathbb D}_{\mcP_n}[0,\infty)$ denote the space of c\'{a}dl\'{a}g functions (right continuous functions whose left limit exits)  from $[0,\infty)$ to $\mcP_n$. The results of \cite{mohle2001classification} concern instances where $(\Psi_{n}^N(t))_{t\ge0}$ has a weak limit  ${\mathbb D}_{\mcP_n}[0,\infty)$ under the Skorohod toplogy $J_1$.}-- see \cite{ethier2009markov} for details.  
 If $(\Psi_n(t))_{t\ge0}$ converges to a continuous time process $(\Psi_{n}^N(t))_{t\ge0}$ in the Skorohod sense, the possible transitions that can occur in the process $\Psi_n(t)$ involve $\sum_{i=1}^rk_i$ of the $n = \sum_{i=1}^rk_i+s$  blocks in $\mcP_n$ merging into $r$ lineages by collapsing into groups of sizes $k_1,\dots,k_r$ while $s$ lineages remain unchanged. Such events are called $(n,k_1,\dots,k_r;s)$-collisions. The rates of these collisions, which uniquely determine  the law of $\Psi_{n}(t)$, will be denoted  by $\lambda_{n,k_1,\dots,k_r;s}$ for $n >2$.

\begin{figure}[h]
\centering
\includegraphics[width=0.9\textwidth]{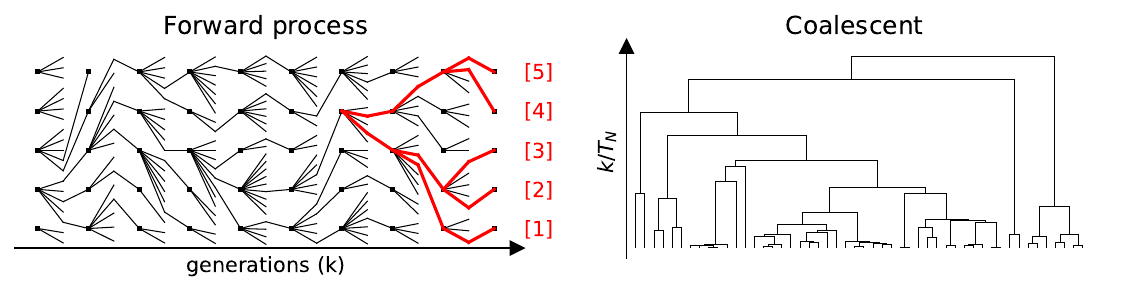}
\caption{(left) A simulation of the Cannings model. Squares indicate individuals at the beginning of each generation and the protruding lines are their offspring. The thick red lines indicate an example of a discrete genealogy, or coalescent $\Psi_{5,k}^5$, obtained from a sample of the final population of labeled cells. For example, $\Psi_{5,1}^5= \{\{1\},\{2,3\},\{4,5\}\}$. (right) A larger simulation of a continuous time coalescent process which is obtained in the limit of the model on the left. }\label{fig:cannings}
\end{figure}

In order to relate the merge rates $\lambda_{n,k_1,\dots,k_r;s}$ to the distribution of $\nu_i$, one notes that after conditioning on $\{\nu_i\}_{i=1}^{N_{\noiseparam}}$ the chance that $r$ groups of sizes $k_1,\dots,k_r$ each descended from individuals $1,\dots,r$ in the previous generation is 
\begin{equation*}\label{eq:nuprod}
\frac{1}{(N)_{\sum_{i=1}^rk_i}}\prod_{i=1}^r(\nu_i)_{k_i}.
\end{equation*}
It follows that
\begin{equation}\label{eq:lambdabk}
\lambda_{n,k_1,\dots,k_r;0} = \lim_{\noiseparam \to \infty}c_{N}^{-1}N^{r-n}\E\left[\prod_{i=1}^r(\nu_i)_{k_i}\right]. 
\end{equation}
In particular, if this limit exists and $c_{N}^{-1}\to \infty$, the time-rescaled coalescent $\Psi_{n}^{N}(t)$ converges to a process $\Psi_{n}(t)$ with merger rates given by Equation \ref{eq:lambdabk}. The rates for $s>0$ can be obtained via a certain recursive relation discussed in \cite{schweinsberg2001coalescents}.

Of particular relevance for our results are those coalescent processes for which there are multiple mergers, but not simultaneous multiple mergers -- in this case $\lambda_{n,k_1,\dots,k_r,s} =0$ unless $r=1$. Such coalescents are known as $\Lambda$-coalescents and are defined by the merger rates
\begin{equation}\label{eq:lambdank2}
\lambda_{n,k} = \lim_{N \to \infty}c_{N}^{-1}N^{1 - n}\E\left[(\nu_1)_k\right]. 
\end{equation}
Any rates defined in this way will satisfy the consistency condition
\begin{equation}\label{eq:lambda_const}
\lambda_{n,k} - \lambda_{n+1,k}= \lambda_{n+1,k+1},
\end{equation}
which is proved in \cite{pitman1999coalescents}.
The intuition behind Equation \ref{eq:lambda_const} is: The difference between the rate to see $k$ mergers in a sample of size $n$  and $n+1$ is entirely caused by mergers of $k+1$ lineages in the sample of $n+1$, since these look like $k$ mergers when restricted to the sample of size $n$. 
Pitman also proved that any triangular array $\{\lambda_{n,k}\}_{n=1,\dots,\infty,k=1,\dots,n}$ satisfying Equation \ref{eq:lambda_const} has a representation in terms of a positive measure $\Lambda:[0,1]\to \reals_{\ge0}$ via the relation
\begin{equation}\label{eq:pitman}
\lambda_{n,k} = \int_0^1 z^{k-2}(1-z)^{n-k}\Lambda(dz). 
\end{equation}
Thus, there is a correspondence between coalescent process and positive measures on the unit interval.

A simple criteria for the convergence to a $\Lambda$-coalescent is found by noting that if $\lambda_{2,2,2;s} =0$ then there are no simultaneous multiple mergers, since these would contribute to this rate. Therefore, if $c_N \to 0$ and 
\begin{equation}\label{condB}
\lim_{N \to \infty}c_{N}^{-1}\frac{\E[(\nu_1)_2(\nu_2)_2]}{N^2} = 0
\end{equation}
the limit process of the genealogy $(\Psi_{n,k}^{N})_{k\ge0}$ is a $\Lambda$-coalescent \cite{mohle2001classification}. The following proposition, which follows from Propositions 1 and 3 of \cite{schweinsberg2003coalescent}, summarizes these observations.

\begin{prop}\label{prop:lambda}
If $c_{N} \to 0$ and Equation \ref{condB} is satisfied, then as $N \to \infty$, $(\Psi_n^N(t))_{t\ge0}= (\Psi_{n,\lfloor t /c_{N} \rfloor}^{N})_{t\ge0}$ converges (in the Skorohod sense)  to a $\Lambda$-coalescent $(\Psi_{n}(t))_{t\ge0}$ with merger rates $\lambda_{n,k}$ given by Equation \ref{eq:lambdank2} for $k=n$ and Equation \ref{eq:lambda_const} otherwise. 
\end{prop}


\subsection{Known result for $q=1$}\label{sec:powerlaw}

We now return to the Cannings model with offspring distributions given by Equation \ref{eq:hass}. When $q=1$, $\growthfactor_i$ has exponential tails and the offspring sizes, $U_i$, have power law tails
\begin{equation*}\label{eq:powerlaw}
P(U_i>u) = P\left(\growthfactor_i >\frac{1}{\noiseparam}\ln u\right) \sim Cu^{-\alpha_1}
\end{equation*}
for $\alpha_1=1/\noiseparam$ and some constant $C>0$.  
The large $N$ limit of the Cannings model  with power law offspring (and $\noiseparam$ fixed) is covered by the main result of \cite{schweinsberg2003coalescent}, which we have stated below in an abbreviated form 

\begin{thm}\label{thm:power}
Assuming Equation \ref{eq:powerlaw} holds, then as $N \to \infty$ we have
 \begin{itemize}
   \item When $2 \le \alpha_1$, $\Psi_{n,\lfloor t /c_N \rfloor}^N$ converges to the $\Lambda$-coalescent with $\Lambda = \delta$, which is called the Kingman coalescent. 
     \item For $1\le\alpha_1<2$, $\Psi_{n,\lfloor t /c_N \rfloor}^N$ converges a $\Lambda$-coalescent where $\Lambda$ given by a $\beta$ distribution ${\rm Beta}(2-\alpha_1,\alpha_1)$. It follows from Equation \ref{eq:pitman} that the merger rates are
\begin{equation}\label{eq:mergerates-beta}
\lambda_{n,k} = \frac{B(k-\alpha_1,n-k+\alpha_1)}{B(2-\alpha_1,\alpha_1)}
\end{equation}
where $B(a,b) = \Gamma(a)\Gamma(b)/\Gamma(a+b)$ is the beta function. This process is a called a $\beta$-coalescent with parameter $\alpha_1$. 

 \item For $\alpha_1<1$, lineages will coalesce in $O(N^0)$ and hence there is no continuous time limit process for any rescaling of time. We refer to \cite{schweinsberg2003coalescent} for a detailed description of this process.
 \end{itemize}
 \end{thm}

This Theorem is closely related to the Generalized CLT (GCLT) for the sum $S_N$ -- see \cite{hallatschek2018selection,okada2021dynamic}.  Recall that the GCLT tells us there are two critical points for the limit law of $S_N$, one where the CLT breaks down ($\alpha_1 = 2$) and another where the LLN breaks down($\alpha_1 = 1$) - see e.g.  \cite{amir2020thinking,amir2020elementary,nolan2020univariate}. In Theorem \ref{thm:power}, these critical points correspond to the appearance of multiple mergers and a disappearance of any continuous time limit process respectively.

\section{Limit coalescent for large but finite offspring variability }\label{sec:largebutfinite}

\subsection{Scaling assumption}
  Our main result concerns the case where $q>1$. We reiterate that in this case the variance is finite for all $\noiseparam$, and therefore in the large $N$ limit (with $\noiseparam$ fixed) the convergence is to the WFD for the Allele frequencies and Kingman Coalescent for the genealogies. For any fixed $N$, if $\noiseparam$ is large enough the WFD/Kingman models are of course going to provide a very poor approximation. In order to better understand exactly what happens when $\noiseparam$ is large, but finite, we set $N = N_{\noiseparam}$ where $N_{\noiseparam}$ grows with $\noiseparam$ in such a way that there is a well-defined limit process as $\noiseparam \to \infty$. The appropriate scaling is related to the thermodynamic limit of the REM \cite{eisele1983third,ben2005limit}; see Section \ref{sec:rem}. %

To state our scaling assumption, we define the \emph{cumulant generating function},
 \begin{equation*}\label{eq:M}
r_{\noiseparam} = \ln \E[e^{\noiseparam \growthfactor_i}]. 
\end{equation*}  
To simplify some formulas later on, we define
$A_{\noiseparam} \equiv \E[S_{\noiseparam}] = N_{\noiseparam}e^{r_{\noiseparam}}$ with 
\begin{equation*}
S_{\noiseparam} \equiv S_{N_{\noiseparam}} = \sum_{i=1}^{N_{\noiseparam}}e^{\noiseparam \growthfactor_i}. 
\end{equation*} The relationship between $r_{\noiseparam}$ and $W^*(z)$ is crucial for our analysis. Using the Laplace method \cite{de1981asymptotic}, which amounts to evaluating $\E[e^{\noiseparam \growthfactor_i}]$ where the integral attains its maximum,  it can be shown that $r_{\noiseparam}$ is asymptotically the convex conjugates of $W^*$. As a result, these functions are related according to the Legendre transform \cite{touchette2005legendre}: 
\begin{equation*}\label{eq:legendre}
r_{\noiseparam} \sim r_{\noiseparam}^* = \inf_{z \ge0}\{z \noiseparam - W^*(z)\}. 
\end{equation*}
It then follows from Equations \ref{eq:hass} that 
\begin{equation}\label{eq:rdef}
r_{\noiseparam}^* =  \frac{1}{q'}\noiseparam^{q'}
\end{equation}
where $q'$ is the so-called dual exponent $q' = q/(q-1)$.
The equivalence between Equations \ref{eq:rdef} and \ref{eq:hass} is a special case of Kasahara–de Bruijn's exponential Tauberian Theorem \cite{mikosch1999regular}.

It follows from Equation \ref{eq:rdef} that all the moments of $U_i$ grow as $\noiseparam^{q'}$, since 
\begin{equation*}
\E[U_i^m] = \E[e^{m\noiseparam \growthfactor}] = e^{r_{m\noiseparam}} \sim e^{m^{q'}r_{\noiseparam}^*} 
\end{equation*} 
Intuitively, if $\ln N_{\noiseparam}$ grows slower than $\noiseparam^{q'}$ then $N_{\noiseparam}$ cannot keep up with the variation as $\noiseparam$ increases and no continuous time limit process will exists.  This serves as motivation for the scaling assumption,
\begin{equation}\label{eq:scale_general}
N_{\noiseparam} = e^{\tau r_{\noiseparam}}.
\end{equation}
Here, $\tau$ is a control parameter closely related to temperature in the Random Energy Model. Limit theorems for the sum $S_{\noiseparam}$ under this scaling assumption are proved in \cite{ben2005limit}. The authors show that, much like in the GCLT, $S_{\noiseparam}$ has two critical points, one where the CLT breaks down and one where the LLN breaks down. Moreover, the limit law of $S_{\noiseparam}$ is scale-invariant (see Theorem \ref{thm:be} in Section \ref{sec:rem}), which strongly suggests a connection between the limit coalescent processes obtained from offspring distribution of the form \ref{eq:Uiexp} under the scaling assumption \ref{eq:scale_general}, and those described by Theorem \ref{thm:power}.

\subsection{Main result}

The following theorem generalizes Theorem \ref{thm:power} to the case $q>1$. Interestingly, we find that the $\beta$-coalescent emerges universally from exponentially large offspring variation.

 \begin{thm}\label{thm:main}
Assume Equations \ref{eq:hass} and   \ref{eq:scale_general}  holds with $q>1$. Let 
\begin{equation}\label{eq:alphaq}
\alpha_{q} = (W^*)'\left((1+\tau)r_1^*\right) = \left(\frac{1+\tau}{q'}\right)^{q/q'}.
\end{equation}
and $\Psi_{n,k}^{\noiseparam} =  \Psi_{n,k}^{N_{\noiseparam}}$. 
Then 
\begin{itemize}
\item For $1<\alpha_q<2$, as $\noiseparam \to \infty$ the discrete coalescent process $(\Psi_{n,\lfloor t /c_{\noiseparam}\rfloor}^{\noiseparam})_{t\ge0}$ converges to a $\beta$-coalescent with parameter $\alpha_q$. 
\item For $\alpha_q >2$, $(\Psi_{n,\lfloor t /c_{\noiseparam}\rfloor}^{\noiseparam})_{t\ge0}$ converges to the Kingman coalescent. 
\end{itemize}
\end{thm}

 In Figure \ref{fig:transition}, the region $1<\alpha_q<2$ is plotted for various values of $q$. By moving upwards through these regions that we obtain $\beta$-coalescents. As expected, for larger $q$, the region becomes more tilted towards the right, since a larger value of $\noiseparam$ is needed to obtain a continuous-time limit process for the same $N$. As $q\to 1$, the region becomes a vertical strip between $1/\noiseparam=1$ and $1/\noiseparam = 2$, hence Theorem \ref{thm:power} is retrieved in this limit.  The regime $\alpha_q \le 1$ requires a different treatment not covered by this result, but is closely related to calculations of the Gibbs measure for the REM at low temperature, as we discuss in Section \ref{sec:rem}.

\begin{figure}[h]
\centering
\includegraphics[width=0.9\textwidth]{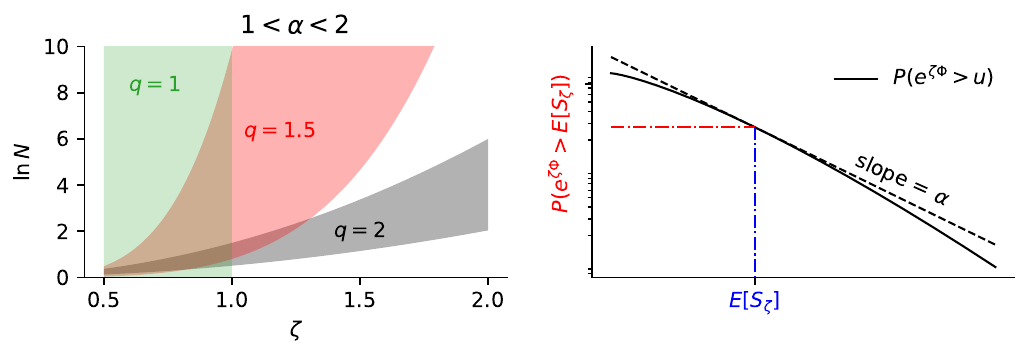}
\caption{(left) The region $1<\alpha<2$ in $\ln N-\noiseparam$ space for different values of $q$.    (right) A diagram of the idea behind the definition of $\alpha$.  }\label{fig:transition}
\end{figure}

We will Theorem \ref{thm:main} in Section \ref{sec:proofs}. Here, we provide an outline of the derivation which we will be useful when making the comparison to results for the REM. Following \cite{schweinsberg2003coalescent}, note that we can replace $\nu_1$ with $NU_1/S_N$ when computing averages (this is justified by Lemma \ref{lem:nu1tailbounds}). Therefore
\begin{equation}\label{eq:numoment_approx}
\E\left[(\nu_1)_k\right] \sim \E\left[\nu_1^k\right] \sim N_{\noiseparam}^k\E\left[\left(\frac{U_1}{S_{\noiseparam}}\right)^k\right].
\end{equation}
A LLN for $S_{\noiseparam}$ (Lemma \ref{lem:LLN}) allows us to replace $S_{\noiseparam}$ with $U_1 + N_{\noiseparam}e^{r_{\noiseparam}}$ (this is justified by Lemma \ref{lem:replaceratio}). Hence, letting $f_U(u)$ denote the density of $U_i$ and replacing sums with integrals (this is justified rigorously by \cite[Lemma 12]{schweinsberg2003coalescent}), we have
\begin{eqnarray}\label{eq:numoments_eq1}
N_{\noiseparam}^k\E\left[\left(\frac{U_1}{S_{\noiseparam}}\right)^k\right] &\sim N_{\noiseparam}^k\int_1^{\infty}\frac{u^k}{(u+Ne^{r_{\noiseparam}})^k}f_U(u)du\\
 &=e^{r_{\noiseparam}}N_{\noiseparam}^{k+1}\int_{(A_{\noiseparam})^{-1}}^{\infty}\frac{z^{k}}{(z +1)^k}f_U(zA_{\noiseparam})dz.
 \end{eqnarray}
To obtain the last equation we have changed variables $z = u/A_{\noiseparam}$. 
If $f_U$ is a power law, then integral can be evaluated and is given in terms of $\Gamma$-functions -- this is one way to prove Theorem \ref{thm:power}, although a different approach is taken in \cite{schweinsberg2003coalescent}.   The essence of our argument is that when we evaluate this integral we can replace $f_U(zA_{\noiseparam})$ with $e^{W^*(\noiseparam^{-1}(\ln A_{\noiseparam} z))}$ and then neglect the higher order terms in 
\begin{eqnarray*}
W^*\left(\frac{\ln A_{\noiseparam} z}{\noiseparam}\right) &= W^*\left(\frac{\ln A_{\noiseparam}}{\noiseparam}\right) + \frac{1}{\noiseparam}(W^*)'\left(\frac{\ln A_{\noiseparam} }{\noiseparam} \right)\ln z + \cdots\\
&= \frac{(1+\tau)^q}{q(q')^q}\noiseparam^{(q'-1)q} +  \left(\frac{1+\tau}{q'}\right)^{q/q'}\noiseparam^{(q'-1)(q-1)-1}  \ln z + \cdots.\\
&= \frac{(1+\tau)^q}{q(q')^q}\noiseparam^{q'} + \alpha_q \ln z  + \cdots.
\end{eqnarray*}
Here, we have used the relations 
\begin{equation*}
(q-1)(q'-1) = 1,  \quad (q'-1)q = q'
\end{equation*}
to simplify the exponents of $\noiseparam$. 
Since the coefficient of $\ln z$ is independent of $\noiseparam$, $f_U$ is approximately a decaying power law with exponent $\alpha_q$, defined by  Equation \ref{eq:alphaq}.  Making this replacement and evaluating the integral leads to the following Lemma, 
which says that $\E\left[(\nu_1)_k\right]$ is dominated by the event $U_1>N$ when $k>\alpha_q$.

\begin{lem}\label{lem:numoments} 
For $1<\alpha_q<k$, $k \in {\mathbb N}$ we have
 \begin{equation}\label{eq:nu1moments}
  \E[(\nu_{1})_k]  \sim  \alpha_q B(\alpha_q,k-\alpha_q)N_{\noiseparam}^k \bbP(e^{\noiseparam \growthfactor_1}> A_{\noiseparam}),\quad \noiseparam \to \infty.
   \end{equation} 
\end{lem}

If $\alpha_q<2$, then from Lemma \ref{lem:numoments} we have  
\begin{equation*}\label{eq:cNtail_q1}
c_{\noiseparam} \sim \alpha_q B(\alpha_q,2-\alpha_q)N_{\noiseparam} \bbP(e^{\noiseparam \growthfactor_1}> A_{\noiseparam})
\end{equation*}
and then from Equation \ref{eq:lambdank2} and another application of Lemma \ref{lem:numoments},
\begin{equation*}\label{eq:lambdann_beta}
\lambda_{n,n} = \frac{B(n-\alpha_q,\alpha_q)}{B(2-\alpha_q,\alpha_q)},\quad n\ge2.
\end{equation*}
These are precisely the merger rates of the $\beta$-coalescent and due to the consistency condition \ref{eq:lambda_const}, uniquely determine the rates $\lambda_{n,k}$ for $k<n$.

On the other hand, when $\alpha_q > k$, $\E[\nu^k]$ is no longer dominated by the tail and we can replace $S_{\noiseparam}$ with $A_{\noiseparam}$ in Equation \ref{eq:numoment_approx}, leading to
\begin{equation*}
\E[\nu_1^k] \sim \frac{N_{\noiseparam}^k\E[e^{k\noiseparam \Phi_1}]}{A_{\noiseparam}^k}. 
\end{equation*}
Thus if $\alpha_q> 2$, 
  \begin{equation}\label{eq:ckingman}
c_{\noiseparam} \sim \frac{N_{\noiseparam}}{A_{\noiseparam}^2}\E[e^{2\noiseparam \growthfactor_1}] \sim e^{(2^{q'}-\tau-2)r^*_{\noiseparam}}. 
 \end{equation}
In this regime, $c_{\noiseparam}$ decays slower than $N^{1-n}\E[\nu_1^n]$ for all $n>2$ and hence all the $\lambda_{n,n}$ except $\lambda_{2,2}$ vanish. 
 

\subsection{Relationship to the result of \cite{cordero2022lambda}}\label{sec:dormancy}

\Rv{
We now remark on the relationship between Theorem \ref{thm:main} and the results of  \cite{cordero2022lambda}. In their model, it is assumed that at the beginning of each generation (referred to as the spring in \cite{cordero2022lambda}), individuals experience a period of dormancy during which no reproduction occurs. At random times, individuals awaken from dormancy and reproduce according to a Yule process -- that is, new individuals are spawned from existing ones at a (deterministic) rate until the end of the generation.  The authors also allow for a period (called the summer) during which all individuals are awake and reproducing, although we will neglect this for the present discussion.   }

\Rv{To make the connection to our model, let $\noiseparam$ rate at which cells divide after the period of dormancy has ended and let $\Phi_i$ denote the duration of the $i$th cell's growth phase (i.e. the difference between the total time between bottlenecks and the dormancy period). In the dormancy model, it follows from properties of Yule processes that $U_1|\Phi_1$ is a geometric random variable with parameter $e^{-\noiseparam \Phi_1}$, hence 
\begin{equation*}
P\left(U_1 >e^{\noiseparam \phi} |\Phi_1\right)= \left(1 - e^{-\noiseparam \Phi_1} \right)^{\lfloor e^{\noiseparam \phi}  \rfloor}
\end{equation*}
and $E[U_1|\Phi_1] = e^{\noiseparam \Phi_1}$. In contrast, in our model $U_1$ is deterministic after conditioning on $\Phi_1$. This distinction should have no effect on the limit coalescent under our scaling assumption, since $\left(1 - e^{-\noiseparam \Phi_1} \right)^{\lfloor e^{\noiseparam \phi}  \rfloor}$ can be replaced with $1_{\Phi_1>\phi}$ in the large $\noiseparam$ limit, 
and therefore 
\begin{eqnarray*}
P\left( \frac{1}{\noiseparam}\ln U_1 >\phi \right) &= \E\left[P\left(\frac{1}{\noiseparam}\ln U_1 >\phi |\Phi_1\right)\right] \\
& = \E\left[P\left(U_1 >e^{\noiseparam \phi} |\Phi_1\right)\right] \sim P\left(\Phi_1 >\phi \right).
\end{eqnarray*}
At least heuristically, this justifies the replacement of $U_1$ with $e^{\noiseparam \Phi_1}$ when calculating asymptotic behavior of the coalescent process. However, we have neglected variation $U_1|\Phi_1$ in order to simplify the derivations in the present paper. 
}

\Rv{The main results of \cite{cordero2022lambda} concern the limit genealogies and are closely related to ours. Theorem 1.3 concerns the situation where the period of growth after dormancy is exponentially distributed and is therefore very similar to Theorem \ref{thm:power} in the present paper. Their more general result, stated below, characterizes the possible forms of $\Lambda$ which can emerge as limits of the discrete coalescents in the dormancy model. }

\Rv{\begin{thm}[Proposition 1.8 and Theorem 1.7 from \cite{cordero2022lambda}]\label{thm:dorm} Suppose $U_i$ are of the form \ref{eq:Uiexp}. For any $\Lambda$-coalescent that emerges as the limit of $(\Psi_{n,\lfloor  t/c_{\noiseparam} \rfloor}^{\noiseparam})_{t\ge0}$ in the dormancy model described above, $\Lambda$ has the form
\begin{equation*}
\Lambda(dx) = b_0 \delta_0(dx) + b_1 \delta_1(dx) + h(x)dx
\end{equation*}
where $\delta_x$ is the point mass at $x$ and $h$ is a probability density on $(0,1)$ with the representation
\begin{equation}\label{eq:dormh}
h(y) = \frac{y^2}{(1-y)^2}g\left(\frac{y}{1-y} \right)
\end{equation}
for a monotone function $g$ satisfying 
\begin{equation*}
\int_0^{\infty} g(v)(1\wedge v^2)dv<\infty. 
\end{equation*}
\end{thm}
Our main result, Theorem \ref{thm:main}, says that under the scaling assumption \ref{eq:scale_general} the limit coalescent is of the form described by Theorem \ref{thm:dorm} with  $g(v) = v^{-1-\alpha_q}$ and $b_0 = b_1=0$.  }

\section{Interpretation in the REM}\label{sec:rem}

\subsection{Background on the REM}

\Rv{The REM was introduced by Derrida in \cite{derrida1981random} as a toy model of spin glasses. As with other models of magnetic systems, the state space is the $n$-hypercube,  $\mcC_n \equiv \{-1,1\}^{n}$ and each configuration $\sigma \in \mcC_n$ is assigned an energy, $E_{\sigma}$. When in equilibrium with a reservoir of temperature $T$, the steady-state distribution of configurations is given by the Boltzmann distribution $P_{\sigma} \propto e^{-\beta E_{\sigma}}$ where $\beta^{-1} = T$ is the inverse temperature (assuming $k_B =1$ for notational simplicity). The quantity of central interest in (equilibrium) statistical mechanics is the free energy, 
 \begin{equation}
F_n = \ln \sum_{\sigma \in \mcC_n}^{2^n}e^{-\beta E_{\sigma}}.
 \end{equation}
It is said that the thermodynamic limit exists if the limit of $F_n/n$ exists, and by differentiating the \emph{free energy density}, $\psi \equiv -\lim_{n \to \infty} F_n/(\beta n)$, with respect to temperature (or other model parameters) various \emph{thermodynamic} relations are obtained \cite{goldenfeld2018lectures,bovier2006statistical}. 
In spin glass models, $E_{\sigma}$ is itself taken to be a random variable, and the free energy density is computed from the mean free energy, $\E[F_n]$. The hope is always that the thermodynamics emerging from a random energy field are typical of systems with very unstructured energy landscapes.  For mathematical convenience $\{E_{\sigma}\}_{\sigma \in \mcC_n}$ is usually taken to be a Gaussian random field, and the simplest such model is of course found by taking  $E_{\sigma}$ to be iid. In order for the thermodynamic limit to exist in this case, the variance of $E_{\sigma}$ must grow proportional to $n$ -- this is precisely Derrida's REM, which can be understood as the limit of a very rugged energy landscape.  } 

Since there are no correlations between energies, we can abandon the hypercube structure and identify the state space with $[2^n]$, writing $E_i$ for the $i$th energy level. Following previous convections, we suppose that $E_i/n$ has variance $1/2$.  Then, when $\Phi_i$ are Gaussian with unit variance,  by setting \begin{eqnarray*}
n &=  \log_2 N_{\noiseparam}\\
\beta &= \sqrt{\frac{2 \ln 2}{\tau}}\\ 
E_i &= -\sqrt{\frac{n}{2}}\growthfactor_i 
\end{eqnarray*}
we see that the partition function in the REM with these parameters is equal to the total number of offspring in the Cannings model
\begin{equation*}
S_{\noiseparam} = \sum_{i=1}^{N_{\noiseparam}}e^{\noiseparam \Phi_i} = \sum_{i=1}^{2^n}e^{-\beta E_i}.
\end{equation*} 
 Note that the negative sign in front of $E_i$ is inconsequential because we have assumed $\E[\growthfactor_i]=0$ and can simply change the sign in Equation \ref{eq:hass}.

Interestingly, even in this simple model one finds a phase transition,  which occurs at a critical value $\beta_c = \sqrt{2\ln2}$, or in our notation, $\tau=1$ (with $q=2$). The nature of this transition can be understood by examining the average number of configurations for which $E_i \in [n\varepsilon,n(\varepsilon + d\varepsilon)]$, which we denote by $\mcN(d\varepsilon)$ \cite{mezard2009information}. Loosely speaking, $\mcN(d\varepsilon)$ is on the order of \cite{kistler2015derrida}
\begin{equation*}
2^ne^{-n\varepsilon^2} = {\rm exp}\left\{n\left[\frac{\beta_c^2}{2}   -\varepsilon^2\right]\right\}. 
\end{equation*}
This approximation makes sense only when $\varepsilon < \beta_c/\sqrt{2}$, since for large $n$ there are virtually no configurations with $\varepsilon >\beta_c/\sqrt{2}$. With these heuristics, one can identify two regimes for the \emph{entropy density}, $s(\varepsilon) \equiv \lim_{n \to \infty} \frac{1}{n}\ln \mcN(d\varepsilon)$: $s(\varepsilon) = \beta_c^2/2   - \varepsilon ^2$ when $|\varepsilon| \le \sqrt{2}\beta_c$ and $s(\varepsilon) = -\infty$ when $|\varepsilon|>\sqrt{2}\beta_c$. 
Meanwhile, the free energy density can be obtained as 
\begin{equation*}
\psi(\beta) = \lim_{n \to \infty} \frac{1}{\beta n}\int_{-\infty}^{\infty} e^{-\beta n \varepsilon}\mcN(d\varepsilon)  = \lim_{n \to \infty} \frac{1}{\beta n}\ln \int_{-\infty}^{\infty} e^{n(s(\varepsilon) - \beta \varepsilon)}d\varepsilon. 
\end{equation*}
and hence $\psi(\beta)$ and $s(\varepsilon)$ are convex conjugates of each other.\footnote{In the terminology of large deviation theory, $\phi(\beta)$ is the scaled cumulant generating function associated with the normalized energy $\varepsilon$, whose large deviation rate function is $s(\varepsilon)$ \cite{touchette2009large}.}  A short calculation using the Laplace method yields Derrida's result:
\begin{equation}
\psi(\beta) = \left\{ \begin{array}{cc}
- \frac{\beta}{4} -\frac{\beta_c^2}{2\beta}&  \beta \le \beta_c\\
-\ln 2& \beta > \beta_c
\end{array}
\right..
\end{equation}

Intuitively, when $\beta\le\beta_c$ the variation in energies is small enough that the LLN can be applied to the partition function. Indeed, one way to arrive at $\psi(\beta)$ in the high temperature phase is to replace $\E[\ln S]$ with $\ln \E[S]$, since 
\begin{equation*}
\E\left[\sum_{i=1}^{2^n}e^{-\beta E_i}\right]= 2^ne^{\beta^2n/4} = \exp\left\{\frac{\beta_c^2}{2} + \frac{\beta^2}{4}\right\}.
\end{equation*} 
In this regime, we say that the system is self-averaging. 
On the other hand, when $\beta>\beta_c$, the system enters a so-called ``frozen'' phase where the partition function is dominated by the extremal statistical weights. 

%

%

The arguments above can be generalized to the sum $S_{\noiseparam}$ for any $q>1$ -- for example, see \cite{eisele1983third}. Of central importance for the application to coalescent theory is following LLN numbers, which tells us the partition function of the REM (with non-gaussian energy distribution) is self-averaging in the high temperature regime. 
\begin{lem}[Theorem 2.1 from \cite{ben2005limit}]\label{lem:LLN} Let $q>1$ and assume
\begin{equation*}\label{eq:taucond}
\tau>\frac{1}{q-1}.
\end{equation*}
Then for all $\varepsilon>0$, 
\begin{equation*}
\lim_{\noiseparam \to \infty}\bbP\left(\left|\frac{S_{\noiseparam}}{A_{\noiseparam}} - 1 \right|>\varepsilon\right) = 0. 
\end{equation*}
\end{lem}
In the context of the Cannings model, the self-averaging of $S_{\noiseparam}$ allows us to make the approximation used in Equation \ref{eq:numoments_eq1} and ensures there is a continuous time limit process.

\subsection{Limit theorem for the partition function}

A richer mathematical structure to the REM is revealed by a careful study of the fluctuations in the partition function. This is closely related to the GCLT for iid sums over scale-invariant random variables, where one can distinguish between regimes of weak and strong self-averaging. The former refers to the case where there is a LLN, but no CLT for the sum. The precise limit Theorem for Gaussian energy distributions is stated in \cite{bovier2006statistical} and the extension to distributions of the form \ref{eq:hass} can be found in \cite{ben2005limit}. We now state (a slightly watered down) version of their result.

\begin{thm}[Theorem 2.3 from \cite{ben2005limit}]\label{thm:be} Let
 \begin{equation}\label{eq:alpha_til}
\tilde{ \alpha}_q  =  \left(\frac{q}{q'}\tau  \right)^{1/q'}. 
 \end{equation}
and 
\begin{equation}\label{eq:Z}
Z_{\noiseparam}= \frac{S_{\noiseparam}- A_{\noiseparam}}{B_{\noiseparam}}.
\end{equation}
where $B_{\noiseparam}$ is defined by 
\begin{equation}\label{eq:basic}
\tau r^*_{\noiseparam} =W^*\left(\frac{\ln B_{\noiseparam}}{\noiseparam}\right). 
\end{equation}
\begin{itemize}
\item  For $1<\tilde{\alpha}_q<2$, $Z_{\noiseparam}$ converges in distribution to an $\alpha$-stable random variable $Z_{\tilde{\alpha}}$ for which the characteristic function is  
\begin{equation}\label{eq:char}
\varphi_{\tilde{\alpha}}(u) \equiv E\left[e^{iZ_{\tilde{\alpha}_q}u}\right] = \exp \left\{- \Gamma(1-\tilde{\alpha}_q)|u|^{\tilde{\alpha}_q}e^{-i\pi \tilde{\alpha}_q {\rm sgn}(u)/2} \right\}
\end{equation}
where the parameter, $\tilde{\alpha}_q$, is given by Equation \ref{eq:alpha_til}.
\item For $\tilde{\alpha}_q>2$, $U_{\noiseparam}$ obeys a central limit theorem, meaning that it converges to a Gaussian when rescaled by the standard deviation. 
\end{itemize}
  \end{thm}
  
  Briefly, the $\alpha$-stable distributions mentioned in this result are defined by the property that they are equal in distribution to linear combinations of realizations of themselves:  $Z$ is $\alpha$-stable if and only if there are constant $c,d,e$ such that $Z = cZ_1 + dZ_2 + e$ for random two variables $Z_1$ and $Z_2$ equal in distribution to $Z$. The GCLT states that $\alpha$-stable distributions arise as limits of sums of iid random variables whose variances are not necessarily finite  \cite{amir2020thinking,zolotarev1986one}. Hence, our result is playing the role of Theorem \ref{thm:be} for the Cannings model by expanding Theorem \ref{thm:power} to the ``large but finite'' regime.

Much like Theorem \ref{thm:main}, the idea behind Theorem \ref{thm:be} is that in the transition regime ($1<\tilde{\alpha}<2$), the sum will be dominated by the maximum. 
From Equation \ref{eq:scale_general} and Equation \ref{eq:hass}, we have 
\begin{eqnarray*}\label{eq:extremal}
\bbP\left({\rm max}_{1\le i\le N} U_i> u\right) &= 1- \left(1-\bbP\left(e^{\noiseparam \growthfactor}>u\right)\right)^N \\
&\sim N\bbP\left(e^{\noiseparam \growthfactor}>u\right)  \sim e^{\tau r_{\noiseparam}-W^*(\ln u/\noiseparam)}.
\end{eqnarray*}
The left-hand side will approach one as $\noiseparam \to \infty$, so in order to obtain a well-defined limit we need to look at deviations on a scale $B_{\noiseparam}$, which is increasing with $\noiseparam$. To this end, we replace $u$ with $uB_{\noiseparam}$ in the exponent of Equation \ref{eq:extremal} and make the approximation 
\begin{eqnarray}\label{eq:basic_taylor}
&\tau r_{\noiseparam}-W^*\left(\frac{\ln u + \ln B_{\noiseparam}}{\noiseparam}\right)  \approx \tau r_{\noiseparam}-W^*\left(\frac{\ln B_{\noiseparam}}{\noiseparam} \right)-\frac{\ln u}{\noiseparam} (W^*)'\left(\frac{\ln B_{\noiseparam}}{\noiseparam} \right).
\end{eqnarray}
In order for this to have a large $\noiseparam$ limit, the first two terms must cancel, indicating that $B_{\noiseparam}$ should satisfy Equation \ref{eq:basic}.   Since $B_{\noiseparam} = O(\noiseparam^{q'})$, it follows from Equation \ref{eq:basic} that the coefficient of $\ln u$ in Equation \ref{eq:basic_taylor} is independent of $\noiseparam$ and given by 
 \begin{equation}\label{eq:alpha_rem}
\tilde{\alpha}_q = \frac{1}{\noiseparam} (W^*)'\left(\frac{\ln B_{\noiseparam}}{\noiseparam} \right)  \end{equation}
 which simplifies to Equation \ref{eq:alpha_til}.
  This plays the role of $\alpha$ in our analysis of the coalescent process. The two parameters agree only at $\alpha=\tilde{\alpha}=1$, which is the transition to the frozen state; see Figure \ref{fig:alpha}. Interestingly, this implies the regime of weak-self averaging in Theorem \ref{thm:be} ($1<\tilde{\alpha}_q<2$) does not exactly correspond to the regime of multiple merger coalescents in Theorem \ref{thm:main} ($1<\alpha_q<2$), although the two coincide in the limit $q\to 1$. 
 
\begin{figure}[h]
\centering
\includegraphics[width=0.3\textwidth]{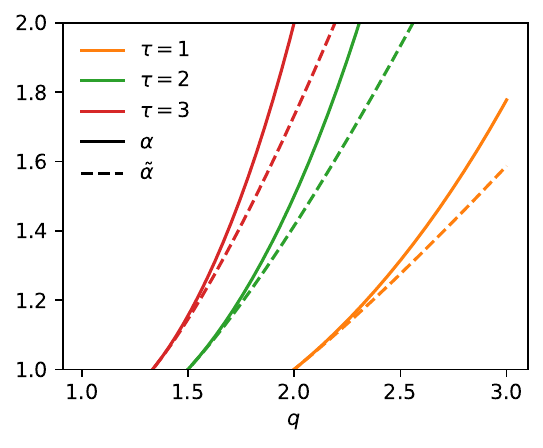}
\caption{$\alpha$ compared to $\tilde{\alpha}$ as a function of the tail exponent $q$ for different values of $\tau$. The lower limit of the plot is the critical point $\alpha=1$ where both expressions agree. Above this point $\alpha>\tilde{\alpha}$ for all $q$.   }\label{fig:alpha}
\end{figure}

\subsection{Gibbs measure}
 
In addition to the partition function, there is an interest in understanding fluctuations in the statistical weights
\begin{equation}
P_{\sigma} = \frac{e^{\noiseparam\growthfactor_{\sigma}}}{S_{\noiseparam}}, 
\end{equation}
in the REM and other disordered systems models.  In the large-$\noiseparam$ limit these weights approach a measure on the infinite dimensional hypercube $\mcC_{\infty} = \{-1,1\}^{\naturals}$ called the \emph{Gibbs measure} -- see \cite{rassoul2015course} for an introduction to this formalism. The question of how the Gibbs measure  fluctuates between replicates of the energies is closely related to the coalescent process, since $\{U_i/S\}_{i=1,\dots,N_{\noiseparam}}$ (asymptotically) have the same distribution as $P_{\sigma}$ when the configurations are projected to the one dimensional lattice $[2^n]=\{1,\dots,2^n\}$. The projected Gibbs measure  approaches the Lebesgue measure on the unit interval in the high temperature regime \cite{bovier2006statistical}.  However, the main focus in this context of the REM appears to have been the low temperature regime, $\tilde{\alpha}_q<1$. 

There is already an established connection between coalescent processes and the Gibbs measure via Derrida's generalized REM (GREM). In this model, the energies are no longer independent, but drawn from a Gaussian random field on $\mcC_n$ whose correlation function depends on a certain ultrametric distance between configurations -- see for a precise description \cite{berestycki2009recent}. This ultrametric distance induces a hierarchical structure to the configuration space. Ruelle gave a mathematical formulation of Derrida's models \cite{ruelle1987mathematical} and the connection to coalescent processes was made in \cite{bolthausen1998ruelle} in the context of their abstract cavity method. The authors show that by sampling configurations from the Gibbs measure and constructing a genealogical tree based on the hierarchies induced by the distance, one obtains (up to a time-change) a continuous time coalescent process now referred to as the Bolthausen-Sznitman coalescent. This processes is nothing but the  $\beta$ coalescent with $\tilde{\alpha}_q=1$.

To understand what Theorem \ref{thm:main} tells us about $P_{\sigma}$, suppose we sample two configurations $i$ and $j$ from the equilibrium distribution and set 
\begin{equation}
\varrho_{i,j} = 1_{\{i=j\}}, 
\end{equation}
which Derrida refers to this as the replica overlap. It is not too difficult to see that $\E[\rho_{1,2}]$ is asymptotic to the inverse coalescent time, $c_{\noiseparam}$: First, notice that conditional on $\{\growthfactor_i\}_{i \in [N]}$ the distribution of $\varrho_{1,2}$ is (see \cite{derrida2021one})
\begin{equation}
\E[\varrho_{1,2}|\{\Phi_{k}\}_{k \in \mcC_n}]  = \sum_{i=1}^{2^n} \left(\frac{e^{\noiseparam \growthfactor_i}}{S_{\noiseparam}} \right)^2
\end{equation}
The joint distribution of the terms in the sum is asymptotic to  that of $\{(\nu_i/N_{\noiseparam})^2\}_{i\in [N]}$. Therefore, recalling Equation \ref{eq:cNdef}, we can see that $\E[\rho_{1,2}] \sim c_{\noiseparam}$. 
It is well known that in the high temperature regime ($\tilde{\alpha}>1$)
\begin{equation}
\E[\varrho_{1,2}]   \to 0,
\end{equation}
while in the low temperature phase Derrida derives an expression in terms of the Beta functions.

Now suppose we sample configurations of the Gibbs measure from $n$ replicates of the REM and consider the event that these configurations are not unique and instead $k_i>1$ come from configuration $i$ for $i=1,\dots,r$ with $\sum_{i=1}^rk_i =n$. Such events are simply the analogues of $(k_1,\dots,k_r;0)$-collisions defined in Section \ref{sec:background_col}, and they occur with probability 
\begin{equation}
N_{\noiseparam}^{r-n}\E\left[\prod_{i=1}^r\left(P_{\sigma} \right)_{k_i}\right]. 
\end{equation}
This is asymptotic to the expression appearing in the definition of $\lambda_{n,k_1,\dots,k_r;s}$ (Equation \ref{eq:lambdabk}) and in the REM is related to the higher order replica overlaps. 
We can then ask what chance of these events is when we take $1/c_{\noiseparam}$ replicates, which yields $\lambda_{n,k_1,\dots,k_r;s}$. Therefore, Theorem \ref{thm:main} tells us there is a phase transition in the typical composition of our samples at the critical point $\alpha=2$. As we have explained above, this transition happens at a lower temperature (higher $\beta$) than the breakdown of the CLT for the partition function at $\tilde{\alpha}_q =2$.

 \section{Proof of Theorem \ref{thm:main}}\label{sec:proofs}
 
In order to simplify notation in the proofs, we set $\alpha = \alpha_q$.

\begin{proof}[Proof of Lemma \ref{lem:numoments}]

The idea of the proof is that most of the contribution to the moments of $\nu_1$ come from the event $\nu_1 > A_{\noiseparam}$. Let $f_{\growthfactor_1}$ denote the density of $\growthfactor_1$; hence 
\begin{equation*} 
\int_{\phi}^{\infty}f_{\growthfactor_1}(x)dx = e^{-W(\phi)}.
\end{equation*}
Define
\begin{equation*}
\eta_{\noiseparam}(z) = \frac{1}{\noiseparam} \ln (A_{\noiseparam}z).
\end{equation*}
By Lemma \ref{lem:replaceratio} and Lemma \ref{lem:nu1ratio} in the \ref{app:lems}, 
\begin{eqnarray*}
E[\nu_1^k] &\sim N^k  \int_{\noiseparam^{-1}\ln(1)}^{\infty} \frac{e^{k\noiseparam \phi}A_{\noiseparam}^{-k}}{(e^{\noiseparam \phi}A_{\noiseparam}^{-1} +1 )^k} f_{\growthfactor_1}(\phi)d\phi\\
&= e^{k\tau r_{\noiseparam}}\int_{A_{\noiseparam}^{-1}}^{\infty}  \frac{z^{k-1}}{(z+1 )^k}\frac{1}{\noiseparam}W'\left(\eta(\noiseparam,z)\right)e^{-W\left(\eta(\noiseparam,z)\right)}dz
\end{eqnarray*}
where we have changed variables $z = u/A_{\noiseparam}$. This is the step which breaks down for $\alpha\le 1$, since Lemma \ref{lem:replaceratio} uses Lemma \ref{lem:LLN}. 

Now define $R_{\noiseparam}(z)$ and $K_{\noiseparam}(z)$ by 
\begin{eqnarray}
W\left(\eta_{\noiseparam}(z)\right)  & =  W\left(\eta_{\noiseparam}(1)\right)  +\hat{\alpha}_{\noiseparam}  \ln z + R_{\noiseparam}\left(z\right)\\
W'\left(\eta_{\noiseparam}(z) \right)  & =\noiseparam\hat{\alpha}_{\noiseparam}  + \noiseparam K_{\noiseparam}(z)
\end{eqnarray}
where
\begin{equation}
\hat{\alpha}_{\noiseparam}= \noiseparam^{-1}W'\left(\eta_{\noiseparam}(1) \right).
\end{equation}
Observe that $\hat{\alpha}_{\noiseparam}\sim \alpha$ and by Equation \ref{eq:hass}
there is a constant $C'$ such that 
\begin{equation*}
K_{\noiseparam}(z) \le  \frac{C'\ln z}{\noiseparam}
\end{equation*}
for large enough $\noiseparam$.

Because $R_{\noiseparam}(z)>0$ for large enough $\noiseparam$ for all $z$, 
\begin{eqnarray}
&\int_{A_{\noiseparam}^{-1}}^{\infty}  \frac{z^{k-1}}{(z+1 )^k}\frac{1}{\noiseparam}W'\left(\eta_{\noiseparam}(z)\right)e^{-W\left(\eta_{\noiseparam}(z)\right)}dz \\
&\quad\le e^{- W(\eta_{\noiseparam}(1))}\int_{A_{\noiseparam}^{-1}}^{\infty}  \frac{z^{k-1-\hat{\alpha}_{\noiseparam}}}{(z+1 )^k}\frac{1}{\noiseparam}W'\left(\eta_{\noiseparam}(z)\right)z. 
\end{eqnarray}
Assume $\noiseparam$ is large enough that $|\hat{\alpha}_{\noiseparam}-\alpha|<\varepsilon$. Then 
\begin{eqnarray}
\int_{A_{\noiseparam}^{-1}}^{\infty}  \frac{z^{k-1-\hat{\alpha}_{\noiseparam}}}{(z+1 )^k}\frac{W'\left(\eta_{\noiseparam}(z)\right)}{\noiseparam}dz
&\le \int_{0}^{\infty}  \frac{z^{k-1-\alpha(\tau)+\varepsilon}}{(z+1 )^k}\frac{W'\left(\eta_{\noiseparam}(z)\right)}{\noiseparam}dz.
\end{eqnarray}
As $\epsilon \to 0$, using that the logarithmic term vanishes, 
\begin{eqnarray}
\int_{0}^{\infty}  \frac{z^{k-1-\alpha+\varepsilon}}{(z+1 )^k}\frac{W'\left(\eta_{\noiseparam}(z)\right)}{\noiseparam}dz &\to \alpha \int_{0}^{\infty}  \frac{z^{k-1-\alpha}}{(z+1 )^k} dz \\
 &= \alpha \frac{\Gamma(\alpha)\Gamma(k-\alpha)}{\Gamma(k)} = \alpha B(\alpha,k-\alpha).
\end{eqnarray}

Similarly, for large enough $\noiseparam$ and any $L$
\begin{eqnarray}
&\int_{A_{\noiseparam}^{-1}}^{\infty}  \frac{z^{k-1}}{(z+1 )^k}\frac{W'\left(\eta_{\noiseparam}(z)\right)}{\noiseparam}e^{-W\left(\eta_{\noiseparam}(z)\right)}dz \\
&\quad\ge e^{- R_{\noiseparam}(L)- W(\eta_{\noiseparam}(1))}\int_{A_{\noiseparam}^{-1}}^{L}  \frac{z^{k-1-\alpha - \varepsilon}}{(z+1 )^k}\frac{W'\left(\eta_{\noiseparam}(z)\right)}{\noiseparam}dz \\
&\quad \sim e^{- R_{\noiseparam}(L)- W(\eta_{\noiseparam}(1))}\alpha \int_{A_{\noiseparam}^{-1}}^{L}  \frac{z^{k-1-\alpha}}{(z+1 )^k}dz
\end{eqnarray}

We can expand $R_{\noiseparam}$ as
  \begin{equation}\label{eq:R_expansion}
  R_{\noiseparam}(L_{\noiseparam})  = \frac{1}{2}\left(\frac{\ln L_{\noiseparam}}{\noiseparam}\right)^2W''\left( \frac{\ln A_{\noiseparam}}{\noiseparam} \right)
  + \frac{1}{6}\left(\frac{\ln L_{\noiseparam}}{\noiseparam}\right)^3W'''\left(\frac{\ln A_{\noiseparam}}{\noiseparam} \right) + \cdots. 
  \end{equation}
Therefore, by the definition of $W^*$ (Equation \ref{eq:hass}), 
  \begin{equation}
 \frac{d^n}{dz^n}W^*(z) = (q-1)\cdots(q-n+1)z^{q-n}
  \end{equation}
  which means the coefficient of $\ln L_{\noiseparam}/\noiseparam$ grows as a power law in $\noiseparam$ with exponent 
  \begin{equation*}
  (q-n)(q'-1) = \frac{q-n}{q-1}.
  \end{equation*} 
Now set $L_{\noiseparam} = B\noiseparam^{q'/4}$ for some constant $B$. The $\noiseparam$ dependence of the $n$th term is 
  \begin{equation*}
n(q'/2-1) + (q'-1)(q-n) = \frac{4-3n}{4}q',
  \end{equation*}
so in the large $\noiseparam$ limit $R_{\noiseparam}(L_{\noiseparam})$ will vanish. In summary,
 \begin{equation}
 \int_{A_{\noiseparam}^{-1}}^{\infty}  \frac{z^{k-1}}{(z+1 )^k}\frac{W'\left(\eta_{\noiseparam}(z)\right)}{\noiseparam}e^{-W\left(\eta_{\noiseparam}(z)\right)}dz 
\sim   e^{- W^*(\eta_{\noiseparam}(1))}\alpha B(\alpha,k-\alpha)
 \end{equation}
Finally, note that 
 \begin{equation}
  e^{- W(\eta_{\noiseparam}(1))} = \bbP\left(e^{\noiseparam \growthfactor_1}> A_{\noiseparam}\right),
 \end{equation}
 so  
 \begin{equation}
  \E[(\nu_{1})_k]  \sim  \alpha B(\alpha,k-\alpha)N_{\noiseparam}^k \bbP(e^{\noiseparam \growthfactor_1}> A_{\noiseparam}).
 \end{equation}

\end{proof}

If $\alpha>k$, the integral derived above diverges and the bounds are no longer useful.  The divergence comes from the very small values of $z$ (meaning small $\growthfactor$ relative to $A_{\noiseparam}$) when we make the linear approximation.  In this case, we have
\begin{eqnarray}
\E\left[\frac{e^{k\noiseparam \growthfactor_1}}{(e^{\noiseparam \growthfactor_1} + A_{\noiseparam})^k}\right] &\sim\frac{1}{A_{\noiseparam}^k} \E[e^{k\noiseparam \growthfactor_1}] \\
&\sim e^{-\left(k(1+\tau)- k^{q'}\right)r_{\noiseparam}^*}.
\end{eqnarray}
Note that $\left(k(1+\tau)- k^{q'}\right)r_1^*\ge W^*((1+\tau)r_1^*)$.  In fact, these are equal exactly at $\alpha_q = 2$.

\begin{proof}[Proof of Theorem \ref{thm:main}]
Lemma \ref{lem:numoments} implies

\begin{equation}
c_{\noiseparam} \sim \alpha B(\alpha,2-\alpha)N_{\noiseparam} \bbP(e^{\noiseparam \growthfactor_1}> A_{\noiseparam}) 
  \sim  \alpha B(\alpha,2-\alpha) e^{-W^*\left(\eta_{\noiseparam}(1)\right) + \tau r_{\noiseparam}^*}.
\end{equation}
Simplifying the exponent, we have
\begin{eqnarray}
W^*\left(\eta_{\noiseparam}(1)\right) + \tau r_{\noiseparam}^* &= W^*\left(\frac{1}{\noiseparam }(1+\tau)r_{\noiseparam}^*\right) + \tau r^*(\noiseparam) \\
&=\left(\alpha^{q'/q} - \frac{1}{q'} - \frac{\alpha^{q'}}{q}\right)\noiseparam^{q'}. 
\end{eqnarray}
The prefactor of $\noiseparam^{q'}$ evaluates to $0$ at $\alpha=1$. Since $q>1$,
\begin{equation}
\frac{d}{d\alpha}\left(\alpha^{q'/q} - \frac{1}{q'} - \frac{\alpha^{q'}}{q}\right) 
= \frac{\alpha^{q'/q} - \alpha^{q'}}{\alpha(q-1)} <0, 
\end{equation}
so that $c_{\noiseparam}^{-1} \to \infty$ for $\alpha>1$. 

The merger rates in Equation \ref{eq:mergerates-beta} follow from Lemma \ref{lem:numoments}, 
so it remains to check \ref{condB}.
Applying Lemma \ref{lem:cNbound} and \ref{lem:numoments},
\begin{eqnarray}
\E[\nu_1^2\nu_2^2] &\sim N_{\noiseparam}^4\E\left[\frac{U_1^2U_2^2}{S_{\noiseparam}^4}\right]  \le  N_{\noiseparam}^4\E\left[\frac{e^{\noiseparam(\growthfactor_1+\growthfactor_2)}}{(e^{2\noiseparam \growthfactor_1} \vee A_{\noiseparam}^2)(e^{2\noiseparam \growthfactor_1}  \vee A_{\noiseparam}^2)}\right] \\
& =  N_{\noiseparam}^4 \E\left[\frac{e^{2\noiseparam \growthfactor_1}}{e^{2\noiseparam \growthfactor_1} \vee A_{\noiseparam}^2}\right]^2 \le (A')^2 N_{\noiseparam}^2c_{\noiseparam}^2
\end{eqnarray}
Notice that we are again using the idea that the expectation is dominated by the event $e^{\noiseparam \growthfactor_1}>A_{\noiseparam}$, since the final expression above is asymptotic to $(A')^2 N_{\noiseparam}^4\bbP(e^{\noiseparam \growthfactor_1}>A_{\noiseparam})^2$. 
It follows that 
\begin{equation}
c_{\noiseparam}\frac{\E[\nu_1^2\nu_2^2]}{N_{\noiseparam}^2} = c_{\noiseparam} \to 0.
\end{equation}

When $\alpha>2$, Lemma \ref{lem:numoments} and Equation \ref{eq:ckingman} implies 
\begin{equation}
c_{\noiseparam}^{-1}\frac{\E\left[\nu_1^3 \right]}{N_{\noiseparam}^2} \sim  \frac{\E\left[\nu_1^3 \right]}{N_{\noiseparam}\E[\nu_1^2]}
=  \frac{N_{\noiseparam}^2\bbP(e^{\noiseparam \growthfactor_1}>A_{\noiseparam})}{\E[\nu_1^2]} \to 0. 
\end{equation}
It follows that $\lambda_{3,3}\to0$, and by Equation \ref{eq:lambda_const} $\lambda_{n,n}=0$ for $n>3$.

 \end{proof}

\section{Discussion}

In this article we have studied the asymptotics of genealogies in the Cannings model when both the population sizes and offspring variation are simultaneously taken to be large. Such limits are not covered by previous results for scale invariant offspring distribution, since in that setting the offspring variation is infinite for finite $N$. Our analysis rests on a certain scaling scaling assumption under which the total number of offspring produced in a generation is equivalent to the partition function of the REM and the offspring numbers $\nu_i$ are related to the Gibbs measure.  As with the REM, competition between fluctuations in growth rates (energies in the REM) with averaging over an increasing system-size (our log population size) leads to a form of weak self-averaging and anomalous scaling of the coalescent time. 

Our main finding is that the $\beta$-coalescent -- a previous studied model of coalescence in populations where variation in offspring is infinite -- also emerge from models where the tail of the offspring distribution is thin, but large fluctuations are not too unlikely.  This is related to the existing limit theorems for the Cannings model proved in \cite{schweinsberg2003coalescent} in the same way that the fluctuation theorem in \cite{ben2005limit} is related to the GCLT for iid sums. Our result does not describe the critical point and low temperature regime, although at least for $q=2$ these can be likely be deduced from previous results on the REM  -- see \cite{bovier2006statistical}.  The limit coalescent processes are the discrete-time $\Xi$-coalescents with simultaneous multiple mergings of lineages described in \cite{mohle2001classification}.

Biologically, Theorem \ref{thm:main} suggests the $\beta$-coalescent serves as a more universal description of neutral evolution in the presence of highly skewed offspring distributions than might be expected. This also indicates that little information about the demographic structure of a population is contained in the coalescent process itself.

\subsection*{Acknowledgments}
I am thankful to Daniel Weissman, Takashi Okada, Linh Huynh, Oskar Hallatschek  and Ariel Amir for helpful discussions and feedback on early versions of this work. Three anonymous referees provided detailed feedback which greatly improved the quality of this manuscript. Finally, I would like to thank the organizers of the conference ``Evolution Beyond the Classical Limits'' held at the Banff International Research Station (BIRS), where preliminary work on this manuscript was carried out.

\appendix
\addtocontents{toc}{\fixappendix}

\section{Additional Lemmas}\label{app:lems}

We will assume $U_i>1$ and therefore $S_N> N$. It is straightforward to obtain the results for the more general case where only Equation \ref{eq:uass} holds. 

\begin{lem}\label{lem:nu1tailbounds}
For $0< \varepsilon <x$
\begin{equation}
\limsup_{\noiseparam \to \infty} N_{\noiseparam}T_{\noiseparam} \bbP\left(\frac{\nu_1}{N_{\noiseparam}}>x\right)\le  \limsup_{\noiseparam \to \infty} NT\bbP\left(\frac{U_1}{S_{\noiseparam}} \ge x - \varepsilon \right)
\end{equation}
and 
\begin{equation}
\limsup_{\noiseparam \to \infty}N_{\noiseparam}T_{\noiseparam} \bbP\left(\frac{\nu_1}{N_{\noiseparam}}>x\right)\ge  \liminf_{\noiseparam \to \infty}N_{\noiseparam}T_{\noiseparam}\bbP\left(\frac{U_1}{S_{\noiseparam}} \ge x + \varepsilon \right)
\end{equation}
\end{lem}

\begin{proof}
Following the proof of  \cite[Theorem 4]{schweinsberg2003coalescent} we condition on $\{X_i\}$ to obtain 
\begin{eqnarray*}
\lim_{\noiseparam \to \infty}NT\bbP(\nu_1/N > x) &= \lim_{\noiseparam \to \infty}NT\E[\bbP(\nu_1/N > x|\{X_i\})]\\
&= \lim_{\noiseparam \to \infty}NT\E[\bbP(\nu_1/N > x|\{X_i\})1_{\{X_1/S\le x-\varepsilon\}}]\\
&\quad\quad\quad+ \lim_{\noiseparam \to \infty}NT\E[\bbP(\nu_1/N > x|\{X_i\})1_{\{X_1/S\ge x-\varepsilon\}}]
\end{eqnarray*}
Using that $\nu_1|\{X_i\}$ has a hypergeometric distribution with parameters $(S,X_1,N)$, 
\begin{eqnarray*}
\E[\bbP(\nu_1/N > x|\{X_i\})1_{\{X_1/S\le x-\varepsilon\}}] &<  \bbP(\nu_1/N > x|X_1 = S(x-\varepsilon))\bbP(X_1/S\le x-\varepsilon)\\
&= \bbP(\nu_1 > Nx|X_1/S  +\varepsilon = x)\bbP(X_1/S\le x-\varepsilon)\\
&\le e^{-2\varepsilon^2N}
\end{eqnarray*}
Since $T$ grows sub-exponentially with $N$, 
\begin{eqnarray*}
\limsup_{\noiseparam \to \infty}NT\bbP(\nu_1/N > x)  &=  \limsup_{\noiseparam \to \infty}NT\E[\bbP(\nu_1/N > x|\{X_i\})1_{\{X_1/S\ge x-\varepsilon\}}]\\
&\le  \limsup_{\noiseparam \to \infty}NT\bbP(X_1/S\ge x-\varepsilon)
\end{eqnarray*}
which is the upper bound. 

The lower bound is obtained similarly using the variable $1_{\{X_1/S\ge x+\varepsilon\}}$. 
\end{proof}

\begin{lem}[Similar to Lemma 6 of \cite{schweinsberg2003coalescent}]\label{lem:nu1ratio}
For $r\ge 1$ and $k_1,\dots,k_r\ge 2$, 
\begin{equation*}
 \E\left[\prod_{i=1}^r(\nu_i)_{k_i}\right]
 \sim N^{\sum_{i=1}^rk_i}\E\left[\frac{e^{\sum_{i=1}^rk_i\noiseparam \growthfactor_i}}{S_{\noiseparam}^{\sum_{i=1}^{r}k_i }}\right]
\end{equation*}

\end{lem}

\begin{proof}
Let $M_{k_1,k_2}^i$ be the event that individuals $k_1$ through $k_2$ descent from individual $i$ in the previous generation
and set 
\begin{equation*}
J_{k_1,\dots,k_r} =\bigcup_{i=1}^{r}M_{\sum_{j\le i}k_j,\sum_{j<i+1}k_j}^i.
\end{equation*} 
As $N \to \infty$ (or $\noiseparam \to \infty$), 
\begin{eqnarray*}
\bbP(J_{k_1,\dots,k_r}) &=  \frac{1}{(N)_{\sum_{i=1}^rk_i}}\E\left[(\nu_1)_{k_1}\cdots (\nu_r)_{k_r}\right]\\
&\sim  \frac{1}{ N^{\sum_{i=1}^rk_i}} \E\left[\prod_{i=1}^r\nu_i^{k_i}\right]
\end{eqnarray*}
Since $S_N \ge N$, 
\begin{eqnarray}
P(J_{k_1,\dots,k_r}) &= \E\left[\bbP\left(J_{k_1,\dots,k_r}\Big|\{\growthfactor_i\}_{i=1}^N\right)\right]= \E\left[\frac{e^{\sum_{i=1}^rk_i\noiseparam \growthfactor_i}}{S_{\noiseparam}^{\sum_{i=1}^{r}k_i }}\right],
\end{eqnarray}
which implies the result. 
\end{proof}

\begin{lem}[Similar to Lemma 8 of \cite{schweinsberg2003coalescent}]\label{lem:cNbound}
For each $k\ge1$ there is a constant $B_k$ such that
\begin{equation*}
\E[\nu^k] \ge B_kN_{\noiseparam}^{k}\E\left[ \frac{e^{k \noiseparam \growthfactor_1}}{e^{k \noiseparam \growthfactor_1} \vee A_{\noiseparam}^k}\right].
\end{equation*}
\end{lem}

\begin{proof}
By Lemma \ref{lem:nu1ratio},
\begin{eqnarray*}
\E[\nu^k] &\sim N^k\E\left[\frac{e^{k\noiseparam \growthfactor_1}}{S_{\noiseparam}^k} \right]\\
&\ge N^k\E\left[\frac{e^{2\noiseparam \growthfactor_1}}{(e^{\noiseparam \growthfactor_1} + 2(A_{\noiseparam}-\mu_{\noiseparam}))^k} 1_{\{\sum_{i=2}^{N_{\noiseparam}}e^{\noiseparam \growthfactor_i} \le 2(A_{\noiseparam}-\mu_{\noiseparam})\}}\right]\\
&=  N^k\E\left[\frac{e^{k\noiseparam \growthfactor_1}}{(e^{\noiseparam \growthfactor_1} + 2(A_{\noiseparam}-\mu_{\noiseparam}))^k} \right]\bbP\left( \sum_{i=2}^{N_{\noiseparam}}e^{\noiseparam \growthfactor_i} \le 2(A_{\noiseparam}-\mu_{\noiseparam})\right)
\end{eqnarray*}
By the law of large numbers
\begin{equation*}
\bbP\left( \sum_{i=2}^{N_{\noiseparam}}e^{\noiseparam \growthfactor_i} \le 2(A_{\noiseparam}-\mu_{\noiseparam})\right) \ge \frac{1}{2}
\end{equation*}
for large enough $\noiseparam$, therefore 
\begin{eqnarray}
 N^k\E\left[\frac{e^{k\noiseparam \growthfactor_1}}{S_{\noiseparam}^k} \right] &\ge \frac{N^k}{2}\E\left[\frac{e^{k\noiseparam \growthfactor_1}}{(e^{\noiseparam \growthfactor_1} + 2(A_{\noiseparam}-\mu_{\noiseparam}))^k} \right]\\
 &\ge \frac{N^k}{2^k}\E\left[\frac{e^{k\noiseparam \growthfactor_1}}{(e^{k\noiseparam \growthfactor_1} + A_{\noiseparam})^k} \right].
\end{eqnarray}
The result follows from
\begin{equation}
(e^{\noiseparam \growthfactor_1} + A_{\noiseparam})^k \le 2^k (e^{k\noiseparam \growthfactor_1} \vee A_{\noiseparam}^k)
\end{equation}
\end{proof}

\begin{lem}\label{lem:replaceratio}
Assume Equation \ref{eq:hass} and Equation \ref{eq:scale_general}. Then for $1<\alpha$
\begin{equation*}
\E\left[\frac{e^{k\noiseparam \growthfactor}}{\left( \sum_{i=1}^{N_{\noiseparam}} e^{\noiseparam \growthfactor_1} \right)^k}\right] \sim \E\left[\frac{e^{k\noiseparam \growthfactor_1}}{\left( e^{\noiseparam \growthfactor_1} +A_{\noiseparam} \right)^k}\right].
\end{equation*}
\end{lem}

\begin{proof}
The idea of the proof is based on  \cite{schweinsberg2003coalescent} combined with  Lemma \ref{lem:LLN}.  By the Lemma \ref{lem:LLN}, for any $\epsilon>0$ and $\delta$ such that $(1-\delta)\E[X_{i,N}]>1$, we can pick $\noiseparam$ large enough that 
\begin{equation}\label{eq:S_N-lln}
\bbP\left(1-\delta\le \frac{S_{\noiseparam}}{A_{\noiseparam}} \le  1+\delta\right) > 1-\varepsilon. 
\end{equation}
If  $\delta$ and $N$ are such that Equation \ref{eq:S_N-lln} holds,  
\begin{eqnarray*}
\E\left[\frac{e^{k\noiseparam \growthfactor}}{\left( \sum_{i=1}^{N_{\noiseparam}} e^{\noiseparam \growthfactor_1} \right)^k}\right]
 &= \E\left[\frac{e^{k\noiseparam \growthfactor}}{\left( \sum_{i=1}^{N_{\noiseparam}} e^{\noiseparam \growthfactor_1} \right)^k}1_{\left\{\sum _{i=2}^{N_{\noiseparam}} e^{\noiseparam \growthfactor_i} < (1-\delta)NA_{\noiseparam}\}\right\}}\right]
\\
&\quad\quad\quad  +  \E\left[\frac{e^{k\noiseparam \growthfactor_1}}{\left( \sum_{i=1}^{N_{\noiseparam}} e^{\noiseparam \growthfactor_i} \right)^k}1_{\left\{\sum _{i=2}^{N_{\noiseparam}} e^{\noiseparam \growthfactor_i} \ge (1+\delta)NA_{\noiseparam}\right\}}\right]\\
  &\le 4 \varepsilon \E\left[\frac{e^{k\noiseparam \growthfactor}}{\left(e^{\noiseparam \growthfactor_1} + N_{\noiseparam} \right)^k}\right]
   +  \E\left[\frac{e^{k\noiseparam \growthfactor_1}}{\left( e^{\noiseparam \growthfactor_1} +(1-\delta) A_{\noiseparam} \right)^k}\right]
\end{eqnarray*}
and 
\begin{eqnarray}
\E\left[\frac{e^{k\noiseparam \growthfactor}}{\left( \sum_{i=1}^{N_{\noiseparam}} e^{\noiseparam \growthfactor_1} \right)^k}\right] \ge (1-\varepsilon)\E\left[\frac{e^{k\noiseparam \growthfactor_1}}{\left( e^{\noiseparam \growthfactor_1} +(1+\delta) A_{\noiseparam} \right)^k} \right].
\end{eqnarray}
The result follows after taking $\varepsilon,\delta \to 0$.
\end{proof}

\section*{References}

\bibliographystyle{alpha}
\bibliography{./largebutfinite.bib}

\end{document}